\documentclass[a4paper,USenglish]{lipics-v2021}

\usepackage{amsmath}
\usepackage{array}
\usepackage{graphicx}

\usepackage{algorithm}
\usepackage[noend]{algpseudocode}

\usepackage{enumitem}
\usepackage{wasysym}
\usepackage{framed}
\usepackage{pgfgantt,rotating}
\usepackage{subfloat}
\usepackage{multirow}
\usepackage{blindtext}

\newtheorem{thm}{Theorem}
\newtheorem{question}{Question}

\nolinenumbers 

\title{Fever: Optimal Responsive View Synchronisation} 

\author{Andrew Lewis-Pye}{London School of Economics \and \url{http://www.lewis-pye.com} }{a.lewis7@lse.ac.uk}{}{}

\author{Ittai Abraham}{Intel}{ittai.abraham@intel.com}{[]}{[]}

\authorrunning{Andrew Lewis-Pye and Ittai Abraham} 

\Copyright{Andrew Lewis-Pye and Ittai Abraham} 



\category{} 

\relatedversion{} 

\acknowledgements{}

\EventEditors{John Q. Open and Joan R. Access}
\EventNoEds{2}
\EventLongTitle{42nd Conference on Very Important Topics (CVIT 2016)}
\EventShortTitle{CVIT 2016}
\EventAcronym{CVIT}
\EventYear{2016}
\EventDate{December 24--27, 2016}
\EventLocation{Little Whinging, United Kingdom}
\EventLogo{}
\SeriesVolume{42}
\ArticleNo{23}

\begin{document}

\maketitle

\begin{abstract}
  \emph{View synchronisation} is an important component of many modern Byzantine Fault Tolerant State Machine Replication (SMR) systems in the partial synchrony model. Roughly, the efficiency of view synchronisation is measured as the word complexity and latency required for moving from being synchronised in a view of one correct leader to being synchronised in the view of the next correct  leader.
    The efficiency of view synchronisation has emerged as a major bottleneck in the efficiency of SMR systems as a whole. A key question remained open: Do there exist view synchronisation protocols with asymptotically optimal quadratic worst-case word complexity that also obtain linear complexity and responsiveness when moving between consecutive correct leaders? 

      We answer this question affirmatively with a new view synchronisation protocol for partial synchrony assuming partial initial clock synchronisation, called \emph{Fever}.  If $n$ is the number of processors and $t$ is the largest integer $<n/3$, then Fever has resilience $t$, and in all executions with at most $0\leq f\leq t$ Byzantine parties and network delays of at most $\delta \leq \Delta$ after $GST$ (where $f$ and $\delta$ are unknown), Fever has worst-case word complexity $O(fn+n)$ and worst-case latency $O(\Delta f + \delta)$.
\end{abstract}

\section{Introduction}
Recent years have seen interest in developing protocols for State Machine Replication (SMR) that work efficiently at scale \cite{cohen2021byzantine}. In concrete terms, this means looking to minimise the latency and the word complexity per consensus decision as a function of the number of participants $n$. Most commonly, this analysis takes place in the partial synchrony communication model, first suggested by Dwork, Lynch, and Stockmeyer \cite{DLS88}. The partial synchrony model forces the adversary to choose a point in time called the Global Stabilisation Time $(GST)$ such that any message sent at time $\mathtt{t}$ must arrive by time $\max\{GST,\mathtt{t}\} + \Delta$. While $\Delta$ is known, the value of $GST$ is unknown to the protocol. This model forms a practical compromise between the synchronous model (where all message delays are bounded by $\Delta$), which is too optimistic, and the asynchronous model (where message delays are finite but unbounded), which is too pessimistic. 

In a recent line of works \cite{HSv1, L22, opodis22} it has been shown that SMR can be solved with optimal resilience and with worst-case word complexity $O(n^2)$ after $GST$. Here, optimal resilience means being able to handle up to $t$  Byzantine faults \cite{DLS88}, where $t$ is the greatest integer less than $n/3$. Given the lower bound of $\Omega(n^2)$ by Dolev and Reischuk \cite{dolev1985bounds}, this bound on word complexity is tight. 

\vspace{0.2cm} 
\noindent \textbf{The optimistic case}. In practical settings, however, one typically cares not only about the worst-case, but also about the complexity and latency in the optimistic case when the actual (and unknown) number of failures $f$ is less than the given bound $t$. Indeed, this is one of the principal motivations for considering the partial synchrony model. In the asynchronous model, where randomness is required after the initial cryptographic setup \cite{fischer1985impossibility}, one can already achieve word complexity which is \emph{expected} $O(n^2)$ per consensus decision \cite{VABA19}. In the partially synchrony model, the hope is that one may be able to define protocols which have worst-case complexity (providing cryptographic assumptions hold) which is $O(fn+n)$. Ideally, such protocols should also be \emph{optimistically responsive}. Roughly, this means that the protocol should function at network speed if it turns out that $f=0$: if $f=0$, the protocol should be live during periods when message delay is less than the given bound $\Delta$, but latency should be a function of the actual (unknown) message delay $\delta$. This is important because the actual message delay $\delta$ may be much smaller than $\Delta$ when the latter value is conservatively set so as ensure liveness under a wide range of network conditions. More formally, we can say that a protocol is optimistically responsive if the latency after $GST$ is $O(\Delta f +\delta)$ -- a precise definition will be given in Section \ref{setup}.  An important feature of this definition is that latency is measured by the time until the first consensus decision \emph{after} GST. Since GST is unknown, the definition applies to multi-shot protocols and cannot be satisfied by single-shot protocols. In particular, being optimistically responsive forces successive correct leaders after GST to complete successive consensus decisions at network speed. 
 
Existing protocols for the partial synchrony model that give optimal resilience and worst-case complexity $O(n^2)$ do not satisfy the desired latency and complexity bounds described above. For such protocols \cite{L22, opodis22}, the worst-case complexity is $O(n^2)$ but not $O(fn+n)$, while latency is $O(n\Delta)$.   

\vspace{0.2cm} 
\noindent \textbf{The bottleneck is view synchronisation}.
 Protocols for Byzantine Agreement and SMR typically divide the instructions into \emph{views}, each with a dedicated leader that coordinates the protocol execution during that view. Since Hotstuff \cite{yin2019hotstuff} shows how to achieve linear complexity within views, the remaining task is to define an efficient protocol that coordinates processors to execute instructions for the same view at the same time as each other. Accordingly, the task of defining efficient protocols for view synchronisation has become a principal focus \cite{Cogsworth21,NK20,L22, opodis22,disc22}, e.g.\ the protocols mentioned above, that achieve worst-case complexity $O(n^2)$ for Byzantine Agreement in the partial synchrony model, achieve this task by defining an appropriate method of view synchronisation.

\vspace{0.2cm} 
\noindent \textbf{Clock assumptions}. 
Dwork, Lynch, and Stockmeyer \cite{DLS88} study several variants of partial synchrony.  In the synchronous communication model, there is a known bound  $\Delta$ on message delay.  In the partially synchronous communication model, the known bound $\Delta$ on message delay only holds after the unknown time GST. Similarly, in the \textit{synchronous processors} model, there is a known bound $\Phi$ on clock drift between correct processors. In the \textit{partially synchronous processors} model,  there is a known bound  $\Phi$ on clock drift between correct processors after GST. Moreover, in the \textit{completely synchronous processors model}, the clocks of all correct processors start at time 0 and there is no clock drift.  

The setting we consider in this paper is the partially synchronous communication model with completely synchronous processors (as studied in Section 4 of \cite{DLS88}). In fact, our model allows for a slight relaxation and some clock drift (see below).\footnote{As explained later, to obtain the results described here it suffices to assume the (potentially realistic) condition that there is some known bound on the different times at which correct processors begin the protocol execution and some known bound on clock drift for correct processors during periods of asynchrony.} Our results do not hold if arbitrary clock drift can occur before GST, but, as described later,  follow-up work already shows that our techniques can be used to improve state-of-the-art results under the assumption that arbitrary clock drift \emph{can} occur prior to  GST. 

From a practical perspective, we believe our result applies to a model that is realistic (at least in some scenarios of interest) and provides a  compelling tradeoff: We show that, by using today's highly reliable hardware clocks, one can achieve view synchronisation with optimal communication complexity and latency, and for both the worst case and optimistic cases. 

From a theoretical perspective, obtaining the positive results of this paper without strong clock synchronisation assumptions remains a challenging open question. We believe the fact that this has not been obtained to date, despite multiple publications and the centrality of this problem, may indicate that there is a natural barrier, and that in fact some type of clock synchronisation is required to obtain our results.

Our formal requirement on clock synchronisation, referred to as \emph{partial initial clock synchronisation},  is defined in Section \ref{setup} and is, in fact, a relaxation of the partially synchronous communication model with completely synchronous processors.  Using the assumption of partial initial clock synchronisation, we are able to define an innovative view synchronisation protocol in which the correct processors send at most $2n$ messages (combined) per view, and which is efficient in both the worst and optimistic cases.

\vspace{0.2cm} 
\noindent \textbf{The result}.
All terms in Theorem \ref{t1} will be formally defined in Section \ref{setup}. Roughly, the worst-case word complexity of a view synchronisation protocol is the maximum number of words (each of maximum length determined by a security parameter) that need to be sent by correct processors during synchrony to synchronise all correct processors on a view with a correct leader. Similarly, the worst-case latency is the maximum time one has to wait during synchrony before all correct processors synchronise on a view with correct leader.

\begin{thm} \label{t1} 
Consider the partial synchrony model with maximum delay $\Delta$ after GST and with partial initial  clock synchronisation. If $t$ is the largest integer less than $n/3$, there exists a view synchronisation protocol with resilience $t$, such that for all executions with at most $0\leq f\leq t$ Byzantine parties and network delays of at most $\delta \leq \Delta$ after $GST$ (where $f$ and $\delta$ are unknown):
\begin{enumerate}
    \item The worst-case word complexity is $O(fn+n)$;
    \item The worst-case latency is $O(\Delta f + \delta)$.
\end{enumerate}

In particular, for $f=0$ this means $O(n)$ complexity and $O(\delta)$ latency and for $f=t$ this is $O(n^2)$ complexity and $O(\Delta n)$ latency.
\end{thm}

 Theorem \ref{t1} obtains worst-case quadratic communication, constant latency per malicious processor, and responsiveness between consecutive honest leaders. This resolves the main open question raised in Cogsworth\cite{Cogsworth21}.

 Combined with Hotstuff, Theorem \ref{t1} gives an optimally resilient SMR protocol for the partial synchrony model that: 
\begin{enumerate} 
\item[(i)] In the worst-case, requires $O(fn+n)$ words to be sent by correct processors after $GST$ before confirmation of the first block of transactions after $GST$, and; 
\item[(ii)] Produces a first confirmed block of transactions after $GST$ within time $O(\Delta f +\delta)$ of $GST$. 
\end{enumerate}

\noindent Since $GST$ is unknown to the protocol, note that similar bounds then hold for the word complexity and latency between honestly produced confirmed blocks after $GST$.


\vspace{0.2cm} 
\noindent \textbf{Implications for results under standard clock assumptions}.
While Fever assumes `partial initial clock synchronisation' to achieve these optimal results, in fact, it has recently been shown\footnote{See \url{https://blog.chain.link/optimal-latency-and-communication-smr-view-synchronization/} } that Fever can also be combined with techniques introduced in \cite{L22}, to  produce a protocol named \emph{Lumiere} that improves on the state-of-the-art \emph{without} the requirement  for partial initial clock synchronisation (but does not achieve the same results as Fever, described above). The results obtained by Lumiere and the trade-offs with Fever are discussed in Section \ref{rw}.

\subsection{Related work} \label{rw}

Tendermint \cite{tendermint16} showed how to use constant size messages for view-change. Casper FFG \cite{casperffg17} extended this approach to allow pipelining.
Hotstuff \cite{yin2019hotstuff} extended these to define an SMR protocol achieving responsivenes and word complexity $O(n)$ within views, but did not rigorously establish an efficient technique for view synchronisation. 
In response to this, a number of papers have described view synchronisation protocols with different trade-offs. 

Cogsworth \cite{Cogsworth21} and Naor-Keidar \cite{NK20} consider a setup in which leaders are chosen according to successive random permutations of the set of processors. They consider a static and \emph{oblivious} adversary, who must choose $GST$ without knowledge of the sequence of randomly chosen leaders, and which must also choose processors to corrupt at the start of the protocol execution without this knowledge. While we do not need to make use of any randomness (beyond that required for the initial cryptographic setup) to establish Theorem \ref{t1}, we also consider such a setup for the purpose of apples-to-apples comparisons in Table 1 (in the `Expected Latency' and `Expected Complexity' columns). Both Cogsworth and Naor-Keidar achieve expected latency $O(\Delta)$ for such a static adversary, but this bound increases to $O(f^2\Delta +\delta)$ in the case that the adversary is adaptive, i.e. if the adversary can choose which processors to corrupt as the execution progresses (and with knowledge as to their choice of $GST$). The principal improvement of Naor-Keidar over Cogsworth is to decrease the expected complexity from $O(n^2)$ in the case of a static and oblivious adversary to $O(n)$. The expected complexity for Cogsworth becomes $O(fn^2+n)$ in the case of an adaptive adversary, and the worst-case complexity is also $O(fn^2+n)$. The expected complexity for Naor-Keidar becomes $O(f^2n+n)$ in the case of an adaptive adversary, and the worst-case complexity is also $O(f^2n+n)$. For a more detailed discussion of Cogsworth and Naor-Keidar, see the Appendix.

\begin{table}[ht]
\centering 
\resizebox{\textwidth}{!}{%
\begin{tabular}{c c c c c c } 
\hline\hline 
Protocol & Expected & Worst-case & Expected & Worst-case  &  Partial Initial\\
& Latency & Latency & Complexity & Complexity &  Clock Sync \\ 
\hline 
Cogsworth & static adv: $O(\Delta)$ & $O(f^2\Delta +\delta)$ & static adv: $O(n^2)$ & $O(fn^2+n)$ & Not needed \\  [1ex] 
& adaptive adv: $O(f^2 \Delta+\delta )$ & & adaptive adv: $O(fn^2 +n)$ & &  \\
\hline 
Naor-Keidar & static adv: $O(\Delta)$ & $O(f^2\Delta +\delta)$ & static adv: $O(n)$ & $O(f^2n +n)$  & Not needed \\[1ex]
& adaptive adv: $O(f^2 \Delta +\delta)$ & & adaptive adv: $O(f^2n+n)$ & &  \\
\hline 
Lewis-Pye & $O(n\Delta)$ & $O(n\Delta)$  & $O(n^2)$ & $O(n^2)$  & Not needed \\
\hline 
Raresync & $O(n\Delta)$ & $O(n\Delta)$  & $O(n^2)$ & $O(n^2)$  & Not needed \\
\hline 

Fever & static adv: $O(\Delta)$  & $O(f\Delta +\delta)$  & static adv: $O(n)$ & $O(fn+n)$  & Needed \\
(this paper) & adaptive adv: $O(f \Delta +\delta)$ & & adaptive adv: $O(fn+n)$ & &  \\
 [1ex] 
\hline 
\end{tabular}}
\vspace{0.3cm} 
\caption{View Synchroniser Comparisons}
\label{table:PM} 
In Table 1, we assume the \emph{bound} $t$ on the number of Byzantine processors is the largest integer less than $n/3$, so that $t=\Theta(n)$, while $0 \leq f \leq t$ is the \emph{actual} number of Byzantine processors. `Complexity' means `word complexity'. Both latency and word complexity are defined in Section \ref{setup}, as is the `partial  clock synchronisation' condition. We only distinguish explicitly between a static and adaptive adversary when this changes the corresponding bound. 

\end{table}

While the published version of Hotstuff \cite{yin2019hotstuff} did not describe any efficient method for view synchronisation, the original version (posted on the arXiv \cite{HSv1}) did roughly outline an approach to meeting the $O(n^2)$ worst-case complexity bound of Dolev-Reischuk. This approach was made precise and rigorously proved in \cite{L22} and \cite{opodis22}. These papers described view synchronisation protocols which we will refer to as `Lewis-Pye' and `Raresync' respectively. A disadvantage of these protocols over Fever (which we describe in this paper) is that they both have worst-case latency $O(n\Delta)$, as opposed to $O(f\Delta + \delta)$ for Fever, and worst-case complexity $O(n^2)$, as opposed to $O(fn+n)$ for Fever. 

Thus far, we have focused on view synchronisation protocols for the partial synchrony model. It should be emphasized that the stronger efficiency bounds for Fever are achieved via a novel view synchronisation protocol combined with stricter assumptions on initial clock synchronisation (as made precise in Section \ref{setup}). 

It is well-known that protocols in the asynchronous model can achieve \emph{expected} complexity $O(n^2)$ per consensus decision (e.g.\ see \cite{VABA19}). The trade-off when compared with the protocol we present here (when combined with Hotstuff) is that such asynchronous protocols do not require synchronous intervals to be live, but still have complexity $O(n^2)$ in the case that $f=0$.  

In \cite{S21}, Spiegelman describes a (single-shot) protocol for Byzantine Agreement which is designed to operate efficiently under both synchronous and asynchronous conditions. The protocol achieves expected $O(n^2)$ complexity in asynchrony and $O(fn+n)$ in synchrony. In the partial synchrony model, the expected complexity after GST remains $O(n^2)$ in the case that $f=0$ (as opposed to $O(n)$ for Fever). If one isolates the part of the (single-shot) protocol which is designed to function under synchrony, and applies just this protocol in the partial synchrony model, then there is no need for view synchronisation, but the protocol is not then optimistically responsive (e.g.\ successive honest leaders do not produce successive consensus decisions in time $O(\delta)$ during synchrony).   

A recent sequence of papers by Bravo,  Chockler, and Gotsman \cite{disc22,DC22, disc20} describe a modular framework for the analysis of view-based SMR protocols. The aim of those papers is complementary to and different than our aim here. While those authors describe a  general framework and are less concerned with establishing optimal results in terms of complexity and latency (such as those described here), our aim  is to describe a specific view synchronisation protocol achieving state-of-the-art efficiency. An advantage of the approach described by those authors is that it allows for a PBFT-style approach to view change, whereby a single leader may persist until correct processors request a change in leader. By contrast, the view synchronisation protocol we describe here has processors automatically pass through views with different leaders. So, this is at least one sense in which the approach described by Bravo et.\ al.\ is more general than what we describe here. 

\vspace{0.2cm} 
\noindent \textbf{Lumiere}. As noted in the introduction, it has recently been shown that Fever can be combined with techniques from \cite{L22} to give a protocol named Lumiere, which significantly  improves on the state-of-the-art without the need for partial  initial clock synchronisation. Roughly, Lumiere divides the views into sets of views called \emph{epochs} and uses Fever for synchronization within epochs, while requiring a heavier synchronization procedure for movement between epochs.   While the Lewis-Pye \cite{L22} protocol achieves $O(n^2)$ worst-case complexity, a significant drawback is that, even after the first synchronisation point (i.e.\ even after the first time after GST when all correct processors are synchronised on a view with correct leader), a single faulty leader can subsequently cause a delay $O(n\Delta)$ between consensus decisions. Lumiere  overcomes this issue -- while the latency and complexity figures for Lumiere with respect to the measures used in Table \ref{table:PM} (which concern the latency and complexity until the first synchronisation point) are the same as for the Lewis-Pye protocol, Lumiere is optimistically responsive and gives optimal performance after the first synchronisation point, without the need for  partial initial  clock synchronisation. If $\mathtt{t}$ is any time after the first synchronisation point, and if latency is given by the time between $\mathtt{t}$ and  the first consensus decision after $\mathtt{t}$, then Lumiere has latency $O(f\Delta + \delta)$ (while maintaining optimal worst-case complexity $O(n^2)$).

\section{The setup} \label{setup}
We consider a set $\Pi= \{ p_0,\dots, p_{n-1} \}$ of $n$ processors, and let $t$ be the largest integer less than $n/3$. Each processor $p_i$ is told $i$ as part of its input. For the proof of Theorem \ref{t1}, we assume an adaptive adversary that is able to choose at most $t$ processors to corrupt as the execution progresses. A processor that is corrupted by the adversary at any point in the execution is referred to as \emph{Byzantine}, and may behave arbitrarily once corrupted. Processors that are not Byzantine are \emph{correct}. We let $f$ denote the actual number of Byzantine processors. 

\vspace{0.2cm} 
\noindent \textbf{Cryptographic assumptions}. Our cryptographic assumptions are standard for papers on this topic. Processors communicate by point-to-point authenticated channels. We use a cryptographic signature scheme, a public key infrastructure (PKI) to validate signatures, and a threshold signature scheme \cite{boneh2001short,shoup2000practical}.  The threshold signature scheme is used to create a compact signature of $m$-of-$n$ processors, as in other consensus and view synchronisation protocols \cite{yin2019hotstuff}. In this paper, either $m=t+1$ or $m=n-t$.  The size of a threshold signature is $O(\kappa)$, where $\kappa$ is a security parameter, and does not depend on $m$ or $n$.
 We assume a computationally bounded adversary. Following a common standard in distributed computing and for simplicity of presentation (to avoid the analysis of negligible error probabilities), we assume these cryptographic schemes are perfect, i.e.\ we restrict attention to executions in which the adversary is unable to break these cryptographic schemes.

 \vspace{0.2cm} 
\noindent \textbf{Communication}. As noted above, processors communicate using point-to-point authenticated channels. We consider the standard partial synchrony model, whereby a message sent at time $\mathtt{t}$ must arrive by time $\max\{GST,\mathtt{t}\} + \Delta$. While $\Delta$ is known, the value of $GST$ is unknown to the protocol. The adversary chooses $GST$ and also message delivery times, subject to the constraints already defined. 

According to the definition above, messages sent prior to $GST$ may be significantly delayed, but are not lost.  We only use the assumption that messages are not lost, however,  when analysing the word complexity of reaching the \emph{first} synchronisation after $GST$. Our view synchronisation protocol works without this assumption, and the assumption could be dropped if one was to consider complexity measures which are less strict than that we consider here, such as that in \cite{Cogsworth21}.

 \vspace{0.2cm} 
\noindent \textbf{Partial  Initial Clock Synchronisation}. To specify our assumptions on the times at which honest processors begin the protocol execution, let $c(p)$ denote the value of processor $p$'s clock. At any point $\mathtt{t}$ in an execution, let   $T(\mathtt{t}) := \{ c(p): \ p \text{ is correct} \}$. In particular, this means that $T(0)$ is the set of all clock values for correct processors at the start of the protocol execution. 
Our required condition regarding initial clock synchronisation is that, for some known bound $\Gamma$: 

\vspace{0.1cm} 
\begin{enumerate} 
\item[$(\dagger_{\Gamma,0})$]  For any $\mathtt{c}\in T(0)$:

\[ | \{ \mathtt{c}'\in T(0) :\ \mathtt{c}'\geq \mathtt{c}-\Gamma \} | \geq t+1. \] 

\end{enumerate} 
Recall that $t$ is the bound on the number of Byzantine processors. So, the condition above says that, for each correct processor $p$, there are at least $t$ other correct processors whose clocks (may be arbitrarily ahead of $p$'s clock but) are at most $\Gamma$ behind $p$'s clock. Specifically, this is really a condition on the most advanced clock of an correct processor. If $p$'s clock is the most advanced amongst correct processors, then we require that there are at least $t$ correct processors whose clocks are at most $\Gamma$ behind $p$'s clock. Another way of looking at this is that all correct processors begin the protocol execution with their local clock set to 0, and that if $p$ is the first correct processor to begin the protocol execution (while other clocks may still be negative, so that those processors are still waiting to start), then at least $t$ other correct processors begin the protocol execution within time $\Gamma $. Note that this condition does not place any bound on the maximum difference between the clocks of correct processors. 

For the sake of simplicity, we will also initially assume that all correct processors have identical clock speeds. Then, in Section \ref{clocksagain}, we will consider realistic relaxations of this condition that suffice to give our results. 

 \vspace{0.2cm}
\noindent   \textbf{The underlying protocol}. We suppose view synchronisation is required for some underlying protocol (such as Hotstuff) with the following properties: 
\begin{itemize} 
\item \textbf{Views}. Instructions are divided into views. Each view $v$ has a designated \emph{leader}, denoted $\mathtt{lead}(v)$.   For some parameter $k\geq 3$ (which can be chosen to suit the protocol designer's needs), we suppose views are grouped into sets of $k$,  so that the leader\footnote{These assumptions are made for the purpose of proving Theorem \ref{t1}. In verifying the bounds given in Table 1, we will also consider the possibility of random leader selection.} for view $v$ is processor $p_i$ where $i:=\lfloor v/k \rfloor \text{ mod } n$. If $v \text{ mod } k=0$, then $v$ is called `initial'.

\item \textbf{Quorum certificates}. The successful completion of a view $v$ is marked by all processors receiving a \emph{Quorum Certificate} (QC) for view $v$. The QC is a threshold signature of length $O(\kappa)$ (for the security parameter $\kappa$ that determines the length of signatures and hash values) combining $n-t$ signatures from different processors testifying that they have completed the instructions for the view. 
In a chained implementation of Hotstuff, for example, the leader will propose a block, processors will send votes for the block to the leader, who will then combine those votes into a QC and send this to all processors. Alternatively, one could consider a (non-chained) implementation of Hotstuff, in which the relevant QC corresponds to a successful third round of voting. Note that the production of QCs is not a  restrictive assumption, since if it is not satisfied one can easily amend the instructions of the protocol so that it is. 
\item \textbf{Sufficient time for view completion}. We suppose there exists some known $x\geq 2$ such that if $\mathtt{lead}(v)$ is correct, if (the global time) $\mathtt{t}\geq \text{GST}$, and if at least $n-t$  correct processors are in view $v$ from time $\mathtt{t}$ until either they receive a QC for view $v$ or until $\mathtt{t}+x\delta$, then  all correct processors will receive a QC for view $v$ by time  $\mathtt{t}+x\delta$, so long as all messages sent by correct processors while in view $v$ are received within time $\delta \leq \Delta$. For the sake of simplicity, we assume $\Gamma$ from the definition of `partial  clock synchronisation' is equal to $x\Delta$-- if these values differ then one can just take the maximum of the two values.


\end{itemize} 

\vspace{0.1cm} 

\noindent \textbf{The view synchronisation task}. For $\Gamma$ as above, we must ensure: 
\begin{enumerate} 
\item If a correct processor is in view $v$ at time $\mathtt{t}$ and in view $v'$ at $\mathtt{t}'\geq \mathtt{t}$, then $v'\geq v$. 
\item There exists some correct $\mathtt{lead}(v)$ and $\mathtt{t}\geq$ $GST$ such that each correct processor is in view $v$ from time $\mathtt{t}$ until either it receives a QC for view $v$ or until $\mathtt{t}+\Gamma$.
\end{enumerate}
Condition (1) above is required by standard view-based SMR protocols to ensure consistency. Since $GST$ is unknown to the protocol, condition (2)  suffices to ensure the successful completion of infinitely many views with correct leaders. By a \emph{view synchronisation protocol}, we mean a protocol which determines when processors enters views and which satisfies conditions (1) and (2) above.

\vspace{0.2cm} 
\noindent \textbf{Complexity measures}. Our proofs are quite robust to the precise notions of latency and word complexity considered, and will hold for any of the definitions used in previous papers on the topic such as \cite{Cogsworth21,NK20,disc22}. For the sake of concreteness, we fix complexity measures which are as strict as possible, and note that if we were to adopt the more relaxed measures used in \cite{Cogsworth21}, for example, then we could weaken the requirement that messages sent before $GST$ are not lost. 

By a `word', we mean a message of length $O(\kappa)$, where $\kappa$ is the security parameter determining the length of signatures and hash values. We make the following definitions.
 Let $\mathtt{t}^{\ast}$ be the least time $>GST$ at which the underlying protocol has some correct $\mathtt{lead}(v)$ produce a QC for view $v$ (if there exists no such time, set $\mathtt{t}^*:=\infty$). The worst-case word complexity is the maximum number of words sent by correct processors (combined) between time $GST+\Delta$ and $\mathtt{t}^{\ast}$. The worst-case latency is the maximum possible value of $\mathtt{t}^{\ast}-GST$. 

 \vspace{0.2cm} 
\noindent \textbf{Defining optimistic responsiveness}. We do not need to define optimistic responsiveness to establish Theorem \ref{t1}. For the sake of concreteness, however, we can define our view synchronisation protocol to be optimistically responsive if the worst-case latency is $O(f\Delta+\delta)$, where $f$ is the (unknown) number of Byzantine processors and $\delta \leq \Delta$ is the actual (unknown) bound on message delay after $GST$. Note that our latency and complexity measures above concern the time and word complexity until the first consensus decision produced by a correct leader strictly \emph{after} GST, and imply that single-shot protocols cannot be optimistically  responsive. 


\section{The protocol} \label{prot}

Recall that views are grouped into sets of $k$,  so that the leader for view $v$ is processor $\lfloor v/k \rfloor \text{ mod } n$. If $v \text{ mod } k=0$, then $v$ is called `initial'. To synchronise processors, we have a predetermined `clock-time' corresponding to each view:  The clock-time corresponding to view $v$ is $\mathtt{c}_v:=\Gamma v$. 

The rough idea is that, at certain points in the execution (and to satisfy optimistic responsiveness), we have processors instantaneously forward their clock to some clock-time $\mathtt{c}_v$ and enter view $v$. We do this in such a way to ensure that, if $p$ is the correct processor whose local clock is most advanced, then there are always at least $t$ other correct processors whose local clocks are at most $\Gamma$ behind $p$'s clock. This will suffice to ensure correct leaders are able to synchronise all correct processors after $GST$.

\vspace{0.2cm}
\noindent The instructions are defined simply as follows: 

\vspace{0.1cm} 
\noindent \textbf{When processors enter views}. Recall that, at any point in the execution,  $c(p)$ is the value of processor $p$'s clock. If $v$ is initial, then $p$ enters view $v$ when  $c(p)=\mathtt{c}_v$.  If $v$ is not initial, then $p$ enters view $v$ if it is presently in a view $<v$ and it receives a QC (formed by the underlying protocol) for view $v-1$.

\vspace{0.1cm} 
\noindent  \textbf{View Certificates}. When a correct processor $p$ enters a view $v$ which is initial, it  sends a $\mathtt{view}\ v$ message to $\mathtt{lead}(v)$. This message is just the value $v$ signed by $p$. Once $\mathtt{lead}(v)$ receives $t+1$ $\mathtt{view}\ v$ messages from distinct processors, it  combines these into a single threshold signature, which is a view certificate (VC) for view $v$, and sends this VC to all processors.\footnote{It is convenient throughout to assume that when a leader sends a message to all processors, this includes itself. } 

\vspace{0.1cm} 
\noindent \textbf{When processors forward clocks}.   At any point in the execution, if a correct processor $p$ receives a QC for view $v-1$ (formed by the underlying protocol) or a VC for view $v$, and if $c(p)<\mathtt{c}_v$, then $p$ instantaneously forwards their clock to $\mathtt{c}_v$.

\vspace{0.2cm}
Pseudocode for the protocol is given in Algorithm 1.

\begin{algorithm} 
\caption{The instructions for processor $p$.}
\begin{algorithmic}[1]

    \State \textbf{Local variables} 

    \State $c(p)$, intially 0 \Comment{This is the value of $p$'s clock }

    \State $v$, initially 0 \Comment{This is the present view of $p$.}

    \State 

    \State \textbf{Global parameters} 

    \State $n$  \Comment{Number of processors}

    \State $t$   \Comment{Largest integer $<n/3$}

    \State $k:=3$      \Comment{Can take larger values.}

    \State $\mathtt{c}_{v'}:=v'\Gamma$, $v'\in \mathbb{N}_{\geq 0}$ \Comment{Defines clock times}

    \State $\mathtt{lead} (v'):= p_i$ for $i=\lfloor v'/k \rfloor \text{ mod } n$ and  $v'\in \mathbb{N}_{\geq 0}$ \Comment{Specifies leaders}

    \State 

    \State \textbf{Upon} $c(p)==\mathtt{c}_{v'}$ for $v'$ initial 

    \State \hspace{0.5cm} Set $v:=v'$  
    \State \hspace{0.5cm} Send a $\mathtt{view}\ v$ message to $\mathtt{lead}(v)$

    \State 

       \State \textbf{Upon} first seeing a QC for view $v'\geq v$ 

    \State \hspace{0.5cm} Set $v:=v'+1$  
    \State \hspace{0.5cm} If $c(p)<\mathtt{c}_{v'+1}$ set $c(p):=\mathtt{c}_{v'+1}$

    \State 

           \State \textbf{Upon} first seeing a VC for initial view $v'> v$ 

    \State \hspace{0.5cm} Set $v:=v'$  
    \State \hspace{0.5cm} If $c(p)<\mathtt{c}_{v'}$ set $c(p):=\mathtt{c}_{v'}$

    \State 

    \State \textbf{If} $p==\mathtt{lead}(v')$ for $v'\geq v$ \textbf{then} 

    \State \hspace{0.5cm} \textbf{Upon} first seeing $\mathtt{view}\ v'$ messages from $t+1$ distinct processors

    \State \hspace{1cm} Form a VC for view $v'$ and send to all processors

\end{algorithmic}
\end{algorithm}

\vspace{0.2cm}
\noindent \textbf{The informal intuition behind the protocol:} Partial initial  clock synchronisation requires that, at the start of the protocol execution, and if $p$ is the correct processor whose local clock is most advanced, there are at least $t$ other correct processors whose local clocks are at most $\Gamma$ behind $p$'s clock. The protocol above is specified to ensure this condition remains true throughout the execution -- see Section \ref{proofs} for a simple proof. This will be easily seen by checking that the condition can never be violated by the forwarding of clocks. 
Now suppose that $\mathtt{lead}(v)$ is correct and that a correct processor, $p$ say, is the first to enter view $v$ after $GST$ at time $\mathtt{t}$. Our condition on local clocks, described above, means that $t$ other correct processors will also enter view $v$ within a short time. Since $\mathtt{lead}(v)$ only requires $t+1$ signatures to form a VC for view $v$, all correct processors will then receive a VC for view $v$ within a short time. The underlying protocol will then have $\mathtt{lead}(v)$ put together a QC for view $v$.

\section{The proofs} \label{proofs}

It is immediate from the instructions that if a correct processor enters a view $v$ then it cannot subsequently enter any lower view.

Recall that, at any point $\mathtt{t}$ in an execution, $T(\mathtt{t}) := \{ c(p): \ p \text{ is correct} \}$. Our condition for `partial initial clock synchronisation' required that a certain condition $(\dagger_{\Gamma,0})$ holds at the start of the protocol execution. This condition requires that if $p$ is the correct processor whose local clock is most advanced, then at least $t$ other correct processors have clocks that are at most $\Gamma$ behind $p$'s clock. The key to the proof is to show that an analogous condition then holds at all times. 

\begin{lemma} \label{lem1} 
For all $\mathtt{t}$ the following condition holds: 

\vspace{0.1cm} 
\begin{enumerate} 
\item[$(\dagger_{\Gamma,\mathtt{t}})$]  For any $\mathtt{c}\in T(\mathtt{t})$:

\[ | \{ \mathtt{c}'\in T(\mathtt{t}) :\ \mathtt{c}'\geq \mathtt{c}-\Gamma \} | \geq t+1. \] 

\end{enumerate} 
\end{lemma} 

 Before proving Lemma \ref{lem1}, we note that the lemma does \emph{not} place any bound on the maximum difference between the local clocks of correct processors. In fact, even if all clocks are initially perfectly synchronised, the local clocks of two correct processors can move arbitrarily far apart prior to $GST$. Nevertheless, the fact that $(\dagger_{\Gamma,\mathtt{t}})$ holds for all $\mathtt{t}$ will suffice to establish Theorem \ref{t1}. 

 \begin{proof} (Lemma \ref{lem1}) 
Since the local clocks of correct processors only ever move forward, it follows that at any point in an execution, if a correct processor $p$ has already contributed to a QC or a VC for view $v$, then $c(p)\geq \mathtt{c}_v$. 
To prove that $(\dagger_{\Gamma,\mathtt{t}})$ holds for all $\mathtt{t}$, suppose towards a contradiction that  there is a first point of the execution, $\mathtt{t}$ say,  for which there exists some correct processor $p$ such that 
$| \{ \mathtt{c} \in T(\mathtt{t}):\ \mathtt{c}\geq c(p) -\Gamma \}| <t+1$. 
Then $p$ must forward its clock at $\mathtt{t}$. There are two possibilities: 
\begin{enumerate} 
\item $p$ forwards its clock because it receives a VC for some view $v$ with $\mathtt{c}_v>c(p)$.  In this case, there must exist at least one correct processor $p'\neq p$ which contributed to the VC for view $v$. By the choice of $\mathtt{t}$, when $p'$ contributed to the VC at $\mathtt{t}'\leq \mathtt{t}$ we had $ | \{ \mathtt{c} \in T(\mathtt{t'}): \mathtt{c}\geq c(p') -\Gamma \}|\geq t+1$. Since $c(p')\geq \mathtt{c}_v$ when it contributed to the VC, and since $c(p)=\mathtt{c}_v$ at $\mathtt{t}$, at $\mathtt{t}$ we have that $ | \{ \mathtt{c} \in T(\mathtt{t}): \mathtt{c}\geq c(p) -\Gamma \}|\geq t+1$ also, which gives the required contradiction. 
\item $p$ forwards its clock because it sees a QC. In this case, at least $t+1$ correct processors must have contributed to the QC, meaning that their clocks are at most $\Gamma$ behind $p$'s clock, which directly gives the required contradiction.  
\end{enumerate} 
\end{proof}

\begin{lemma} \label{lem2}
If $v$ is initial and $\mathtt{t}$ is the first time any correct processor enters a view $\geq v$:
\begin{enumerate} 
\item[(i)] A correct processor enters view $v$ at $\mathtt{t}$;
\item[(ii)] No correct processor enters any view $v'>v$ at $\mathtt{t}$, and;
\item[(iii)] $c(p)\leq \mathtt{c}_v$ for all correct $p$ at $\mathtt{t}$.
\end{enumerate}
\end{lemma}
\begin{proof} 
 Consider the first time any correct processor $p$ enters a view $v'\geq v$.  It cannot be because $p$ sees a VC for view $v'$, because some correct processor must then have contributed to that VC and already have been in view $v'$. It cannot be because $p$ sees a QC for view $v'-1>v-1$, because $t+1$ correct processors must have already contributed to that QC (noting that $p$ is the first to enter any view $\geq v$).
It follows that the first view $v'\geq v$ entered by any correct processor is $v$. When the first correct processor $p$ enters view $v$ we have  $c(p)=\mathtt{c}_{v}$ (either simply because it reaches this value, or else because $p$ sees a QC for view $v-1$), and that $c(p')\leq \mathtt{c}_{v}$ for all correct $p'$ at this point.
\end{proof}

\begin{definition} 
Let $\mathtt{t}(v)$ be the first time at which a correct processor enters view $v$.
\end{definition}

Since correct processors enter an unbounded number of views, it follows from Lemma \ref{lem2} that if $v$ is initial then $\mathtt{t}(v)\downarrow$ and\footnote{We write $x\downarrow$ to denote that the variable $x$ is defined. } $\mathtt{t}(v')>\mathtt{t}(v)$ whenever $v'>v$ and $\mathtt{t}(v')\downarrow$. 

Note also that if $v$ is initial then, for $j\in (0,k)$, a QC for view $v+j$  cannot be formed prior to the formation of a QC for view $v+j-1$. This  follows because (since $v$ is initial, and for $j$ in the given range) correct processors do not enter view $v+j$ without seeing a QC for view $v+j-1$. 
The next lemma will be used to show that correct processors spend a sufficiently long time in each view that a correct leader after $GST$ will be able to produce QCs.

\begin{lemma} \label{lem3}
Suppose $v$ is initial. For each $j\in [0,k)$, let $\mathtt{s}_j$ be the first time (if there exists such) at which a correct processor sees a QC for view $v+j$. The first time at which any correct processor enters view $v+k$ is the minimum amongst the values  $\{\mathtt{t}(v)+k\Gamma \} \cup \{ \mathtt{s_j} + (k-1-j)\Gamma: \ \mathtt{s}_j \downarrow\}  $.   
\end{lemma}
\begin{proof}
By Lemma \ref{lem2}, some correct processor $p$ enters $v$ at $\mathtt{t}(v)$, and all correct processors $p'$ have $c(p')\leq \mathtt{c}_v$ at this point. As we reasoned in the proof of Lemma \ref{lem2}, it cannot be the case that the first time any correct processor enters a view $v'\geq v+k$ it is because it sees a VC for view $v'$ or a QC for $v'-1\geq v+k$. It follows that the first time a correct processor $p$ enters view $v+k$ it is because its local clock has reached $\mathtt{c}_{v+k}$. 
This happens either because $p$ saw a QC for view $v+j$ ($j\in [0,k)$) and then time $(k-1-j)\Gamma$ passed (meaning zero time if $j=k-1$), or else because  $p$ was the first correct processor to enter view $v$ and time $k\Gamma$ passed since that point.  
\end{proof}

With Lemmas \ref{lem1}, \ref{lem2} and \ref{lem3} in place, the basic intuition behind the idea that a correct leader will produce a QC after $GST$ is clear. Let $\mathtt{lead}(v)$ be correct and such that no correct processor  enters view $v$ prior to $GST$. From Lemma \ref{lem2}, it follows that no correct processor enters any view $v'>v$ prior to $\mathtt{t}(v)$.  By Lemma \ref{lem1},  at least $t+1$  correct processors will have entered view $v$ within time $\Gamma$ -- by Lemma \ref{lem3}, no correct processor will be in any view $>v$ prior to the first of $\mathtt{t}(v)+k\Gamma$ or else the formation of a  QC for view $v$.  The correct processor $\mathtt{lead}(v)$ will then form a VC for view $v$, and all correct processors will be in view $v$ by time $\mathtt{t}+\Gamma +2\Delta$ unless a QC for view $v$ has already been formed by this point. This means that all processors will receive a QC for view $v$ by time $\mathtt{t}+2\Gamma +2\Delta$. Since $\Gamma\geq 2\Delta$ and $k\geq 3$, this suffices. 

Now let us see the details. In the below, we prove more than the fact that a correct $\mathtt{lead}(v)$ will produce a QC for one of the views in $[v,v+k)$. We show that $\mathtt{lead}(v)$ will produce QCs for multiple successive views if $k$ is large enough, since this will be useful in some implementations (such as chained implementations of Hotstuff etc). 

\begin{lemma} \label{lem4}
Suppose $v$ is initial, $\mathtt{lead}(v)$ is correct, and that $\mathtt{t}(v)\geq GST$. Then correct processors will see QCs  for all views in $[v,v+k-2)$ before entering view $v+k$.
\end{lemma}

\begin{proof}  By Lemma \ref{lem2}, no correct processor has entered any view $v'>v$ at $\mathtt{t}(v)$. By Lemma \ref{lem1}, $(\dagger_{\Gamma,\mathtt{t}(v)})$ is satisfied, which means at least $t+1$ $\mathtt{view}\ v$ messages will have been sent to $\mathtt{lead}(v)$  by $\mathtt{t}(v)+\Gamma$ -- by Lemma \ref{lem3}, no correct processor will be in any view $>v$ prior to the point at which these $t+1$ $\mathtt{view}\ v$ messages have been sent to $\mathtt{lead}(v)$. Then $\mathtt{lead}(v)$ will have sent out a VC for view $v$ by $\mathtt{t}(v)+\Gamma +\Delta$, which will be received by all correct processors by time $\mathtt{t}(v)+\Gamma +2\Delta$. It then follows from Lemma \ref{lem3}, and since $\Gamma\geq 2\Delta$, that a QC for each view $j\in [0,v+k-2)$ will be seen by all correct processors by time $\mathtt{t}(v)+\Gamma +2\Delta +(j+1)\Gamma$, prior to any point at which a correct processor enters view $v+k$.  
\end{proof}

 \begin{lemma} \label{lem5}
 The worst-case word complexity is $O(fn+n)$ and the worst-case latency is $O(\Delta f + \delta)$.
 \end{lemma}
 \begin{proof} 
We deal with the word complexity first. Let $p$ be the correct processor whose clock is most advanced at $GST$ (breaking ties arbitrarily). Suppose $p$ is in view $v$ at $GST$. Let $v_0$ be the greatest initial view $< v$ such that $\mathtt{lead}(v_0)\neq \mathtt{lead}(v)$ and $\mathtt{lead}(v_0)$ is correct. Let $v_1$ be the least initial view $>v$ such that $\mathtt{lead}(v_1)$ is correct. Since no correct processor will enter any view $>v_0$ prior to the least of $\mathtt{t}(v_0)+k\Gamma$ or the first time at which a correct processor sees a QC for view $v$, and since $(\dagger_{\Gamma,\mathtt{t}(v_0)})$ is satisfied, $\mathtt{lead}(v_0)$ must have sent a VC for view $v_0$ to all processors prior to $GST$. All correct processors will therefore be in at least view $v_0$ by $GST+\Delta$. Lemma \ref{lem4} shows that all correct processors will see a QC for view $v_1$ before entering view $v_1+k$. Let $f^{\ast}$ be the number of Byzantine leaders for initial views in the interval $(v_0,v_1)$. Correct processors will send a maximum of $2(f^{\ast}+3)n$ many $\mathtt{view}$ messages (combined) between $GST+\Delta$ and the time at which $\mathtt{lead}(v_1)$ produces a QC for view $v_1$. If the underlying protocol has correct processors send $O(n)$ messages per view (e.g.\ Hotstuff), then the underlying protocol will also have correct processors send $O((f^{\ast} +3)n) $ messages during this interval. So the worst-case word complexity is $O(fn+n)$, as required. 

Next, we consider the worst-case latency. If $f>0$, then it suffices to observe that $\mathtt{lead}(v_1)$ will produce a QC for view $v_1$ by time $GST + k(f^{\ast}+3)\Gamma$. So suppose $f=0$ and consider the number $d$ of correct processors in view $v$ at $GST$. If $d\geq t+1$, then $\mathtt{lead}(v)$ will produce a QC for view $v$ within time $O(\delta)$ according to our assumptions on the underlying protocol, unless at least $t+1$ processors enter view $v+k$ before this occurs. In the latter case, $\mathtt{lead}(v+k)$ will produce a QC for view $v+k$ within time $O(\delta)$. If $d<t+1$, then (the previous leader) $\mathtt{lead}(v_0)$ will produce a QC for view $v$ within time $O(\delta)$, unless at least $t+1$ processors enter view $v$ before this occurs. In the latter case, $\mathtt{lead}(v)$ will produce a QC for view $v$ within time $O(\delta)$.
 \end{proof}

 Lemma \ref{lem5}  completes the proof of Theorem \ref{t1}. We finish this section by justifying the entries of Table 1, which state that the expected latency is $O(\Delta)$ and the expected word complexity is $O(n)$ with a static adversary according to the model of \cite{Cogsworth21}. According to this model, leaders are given by successive random permutations of the set of all processors. The adversary is static and \emph{oblivious}, which means that they must choose which processors to corrupt at the start of the protocol execution without knowledge as to the random sequence of leaders, and must also choose $GST$ without this knowledge. In this case, the expected value $f^{\ast}$ from the proof of Lemma \ref{lem5} is $O(1)$, which means we get expected latency $O(\Delta)$ and expected word complexity $O(n)$, as required.

\section{Tying up loose ends} \label{clocksagain}

\subsection{A note on optimistic responsiveness and Byzantine leaders}

In the proof of Lemma \ref{lem5}, it was only actually the number of Byzantine \emph{leaders} before the first correct leader after $GST$ that mattered (rather than the total number of Byzantine parties) in establishing that the worst-case word complexity is $O(fn+n)$. The proof that the worst-case latency is $O(f\Delta + \delta)$ was somewhat more subtle, and considered the total number of Byzantine processors. Roughly, the difficulty occurs when none of the relevant leaders are Byzantine, but a correct processor has just entered initial view $v$ at GST, while all other correct processors are still in previous views. In this scenario, the previous leader cannot produce a QC (if $t$ parties are Byzantine while not in their role as leader). Meanwhile, the leader for view $v$ has to wait time $\Gamma +\delta$ before producing a VC. As we argued in the proof of Lemma \ref{lem5}, this is not an issue if we let $f$ count the total number of Byzantine parties.

Once a first correct leader produces a QC after $GST$, however, this subtlety is no longer relevant. Let $\mathtt{lead}(v)$ be correct and such that $\mathtt{t}(v)\geq GST$.  Let $v'$ be the least view with correct leader $>v$ and suppose that the number of initial views with Byzantine leader in the interval $(v,v')$ is $f^{\ast}$. 
Then the proof of Lemma \ref{lem5} is easily modified to show that correct processors send $O(f^{\ast}n+n)$ many words between the times at which $\mathtt{lead}(v)$ and $\mathtt{lead}(v')$ produce QCs, and that the time between these events is $O(f^{\ast}\Delta +\delta)$.

\subsection{Revisiting the assumptions regarding clock synchronisation}

In Section \ref{setup} we assumed that all processors have identical clock speeds. We now consider to what extent we can relax this condition. As we do so, we also consider how realistic are the required assumptions in the context of reasonable bounds on network delays, the length of periods of asynchrony etc., and in a context where atomic clocks are available for use by processors.  Recall that atomic clocks can reasonably be assumed to have error less than 1 second every 100 million years.\footnote{See, for example, \url{https://en.wikipedia.org/wiki/Atomic_clock}} 

In the partial synchrony model it is only for the sake of technical convenience that we consider a single period of asynchrony and then a single period of synchrony after $GST$. In reality, we are interested in contexts where network conditions oscillate between synchrony and asynchrony. We require our protocols to maintain consistency during periods of asynchrony, and to be live during periods of synchrony. To ensure that our analysis extends to such a scenario, let us therefore consider our requirements as the network oscillates between periods of sychrony and asynchrony in this fashion. 

Fix $k:=3$. A similar analysis will also apply for larger values of $k$. Let us say an open interval $(\mathtt{t},\mathtt{t}')$ is synchronous if every message sent in this interval arrives within time $\Delta$. Let $\ell:= 3(t+3)\Gamma$, where $t$ is the bound on the number of Byzantine processors.  If  $(\dagger_{\Gamma,\mathtt{t}'})$ holds for all $\mathtt{t}'\in I:=(\mathtt{t},\mathtt{t}+\ell)$ and if $I$ is synchronous, the proofs of Section \ref{proofs} established that some correct leader will produce a QC during interval $I$ and send this to all processors. With this in mind, we inductively define a sequence of times $(\mathtt{t}_i)_{i\geq 0}$ as follows: 
\begin{itemize} 
\item Let $\mathtt{t}_0$ be the least that $(\mathtt{t}_0, \mathtt{t}_0 +\ell)$ is synchronous. 
\item Given $\mathtt{t}_i$, let $\mathtt{t}_{i+1}$ be the least $\mathtt{t}\geq \mathtt{t}_i+\ell$ such that $(\mathtt{t}, \mathtt{t} +\ell)$ is synchronous.
\end{itemize}

We suppose that every $\mathtt{t}_i$ is defined. For concreteness, it is also useful to stipulate some specific values -- a similar argument will hold for comparable values: 

\begin{itemize}
\item Suppose $\Delta= 1$ second. 
\item Suppose $\mathtt{t}_0<10^5$ years and the maximum value $\mathtt{t}_{i+1}-\mathtt{t}_{i}$ is less than $10^5$ years. 
\item For view $v$, suppose $\mathtt{t}$ is the first time at which $t+1$ correct processors are in view $v$, and that $\mathtt{t}'$ is the first time at which a correct processor sees a QC for view $v$. Define $\mathtt{u}_v:= \mathtt{t}'-\mathtt{t}$. When $\mathtt{u}_v$ is defined, we suppose it always has at least the mininum value $\mathtt{u}$. For the sake of concreteness, we suppose $\mathtt{u}=10^{-2}$ seconds.  
\item Suppose that $(\dagger_{\Delta,0})$ holds, and that $\Gamma=2\Delta$. 
\end{itemize}
We note that the assumptions above are weak: In particular, we assume only that a synchronous interval exists every $10^5$ years. Then we claim $(\dagger_{\Gamma,\mathtt{t}})$ holds for all $\mathtt{t}$ -- this is the condition required to ensure that every interval $(\mathtt{t}_i, \mathtt{t}_i+\ell)$ has a correct leader produce a QC. Towards a contradiction, suppose there exists a least value $i^{\ast}$ such that $(\dagger_{\Gamma,\mathtt{t}})$ fails to hold for some $\mathtt{t}^{\ast}$ in the interval $[\mathtt{t}_{i^{\ast}},\mathtt{t}_{i^{\ast}+1})$. For each $\mathtt{t}$, let $\Gamma(\mathtt{t})$ be the smallest $\Gamma'$ such that $(\dagger_{\Gamma',\mathtt{t}}) $ holds. Note that: 
\begin{itemize} 
\item Every interval $(\mathtt{t}_i,\mathtt{t}_i +\ell)$ such that $i<i^{\ast}$ has at least one correct processor synchronise the clocks of correct processors to within time $\Delta$, i.e. $(\dagger_{\Delta,\mathtt{t}})$ holds for some $\mathtt{t}$ in this interval. 
\item When a correct processor sees a QC and forwards its clock at $\mathtt{t}$, this may cause $\Gamma(\mathtt{t})$ to increase, e.g.\ if the leader is Byzantine and only sends the QC to certain processors, or if the QC is not sent during a synchronous interval. In this case, however, the maximum value of $\Gamma(\mathtt{t}) $ is still at most $\Gamma-\mathtt{u}$. 
\item If a processor forwards its clock because it sees a VC at $\mathtt{t}$, this does not increase $\Gamma(\mathtt{t})$. 
\end{itemize}

Define $\mathtt{t}:=\mathtt{t}_{i^{\ast}-1}+\ell$ if $i^{\ast}\neq 0$, and define $\mathtt{t}:=0 $ if $i^{\ast}=0$. Let $\mathtt{t}^{\ast}$ be defined as above. We conclude that $\Gamma(\mathtt{t})$ is at most $\text{max} \{ \Delta, \Gamma-\mathtt{u} \}$, to within a small error term which is the maximum drift of clocks within an interval of length $\ell$. Since we suppose $\mathtt{u}=10^{-2}$ seconds, since our clocks have drift at most 1 second every 100 million years, and since some clocks may drift slow while others drift fast,  this means that $\mathtt{t}^{\ast}-\mathtt{t}>5\times 10^5$ years. This gives the required contradiction, since we assumed above that $\mathtt{t}_0<10^5$ years and $\mathtt{t}_{i+1}-\mathtt{t}_i$ is less than $10^5$ years for all $i$.  

\section{Concluding comments}

We have defined Fever, which is a novel view synchronisation protocol. If $n$ is the number of processors and $t$ is the largest integer $<n/3$, then Fever has resilience $t$, and in all executions with at most $0\leq f\leq t$ Byzantine parties and network delays of at most $\delta \leq \Delta$ after $GST$ (where $f$ and $\delta$ are unknown), Fever has worst-case word complexity $O(fn+n)$ and worst-case latency $O(\Delta f + \delta)$. This improves significantly on the state-of-the-art. 

The trade-off is that Fever requires greater assumptions than previous view synchronisation protocols regarding the drift of clocks prior to $GST$. We have argued in Section \ref{clocksagain} that there are scenarios in which our required assumptions are reasonable. Atomic clocks can now be purchased for a few thousand US dollars, and we showed that under reasonable assumptions regarding network latency etc., a system implementing Fever will be able to handle periods of asynchrony of the order of $10^5$ years. Of course, this is more than is reasonably required, and so even the use of less accurate clocks may suffice in many scenarios. 

We also noted, in Section \ref{rw}, that it has been shown that Fever can be combined with other techniques to produce a protocol called Lumiere that improve on the state-of-the-art without the need for partial  initial clock synchronisation. Since Lumiere still has trade-offs with Fever, the following question remains: 

\begin{question}
Does there exists a view synchronisation protocol for the partial synchrony model that achieves the same efficiency bounds as Fever, but which can accommodate unbounded clock drift prior to $GST$? 
\end{question}

\bibliographystyle{plainurl}

\section{Appendix} 
The following fairly nuanced observation regarding Cogsworth and Naor-Keidar came out of a private conversation with some of the authors of those papers. For the purpose of this discussion, we assume some familiarity with both of those papers \cite{Cogsworth21,NK20}.

 In both Cogsworth and Naor-Keidar,  there is a certain known bound, which we will call $\delta^{\ast}$ here, and which determines the delay before a correct processor gives up on one relay and tries the next. In those papers, $\delta^{\ast}$ is set to equal $2\Delta$, which gives the results described in Table 1. By playing with different values for $\delta^{\ast}$, however, it is possible to achieve a trade-off between worst-case latency and expected complexity. While the corresponding results do not significantly impact the narrative of this paper, they may be of interest to the dedicated reader.

In the following discussion, we will consider only the case of a static adversary. We assume that $\delta^{\ast}$ may be much smaller than $\Delta$, but is at most $2\Delta$. For Cogsworth, the possibility that $\delta^{\ast}$ may be much smaller than $\Delta$ now means that the worst-case latency is $ O(f^2\delta^{\ast} +f\Delta +\delta)$, while the worst-case complexity remains $ O(fn^2+n)$. Expected latency is $O(\Delta)$ (and $O(\delta)$ if $f=0$).

The issue with making $\delta^{\ast}$ very small is that, while this decreases the worst-case latency, it increases the \emph{expected} complexity in the case that $\delta >\delta^{\ast}$. In this case, the expected complexity becomes $O(n^2)$ even with benevolent faults (the standard version of Cogsworth has expected complexity $O(n)$ in the case of benevolent faults).  

For Naor-Keidar, if actual latency is at most $\delta^{\ast}$, then worst-case latency is $O(f^2\delta^{\ast} +f\Delta +\delta)$ and worst-case communication is $O(f^2n+n)$. Expected latency is $O(\Delta)$ (and $O(\delta)$ if $f=0$). The expected communication cost is $O(n)$.
If actual latency is $>\delta^{\ast}$, then the worst-case latency is $O(f^2\delta^{\ast} +f\Delta +\delta)$ and worst-case communication is $O(n^2f+n)$. Expected latency in this case is $O(\Delta)$ and expected communication is $O(n^2)$. So, again, there is a trade-off. Setting a small value of $\delta^{\ast}$ decreases the worst-case latency, but increases the expected communication in the case that $\delta>\delta^{\ast}$.
\end{document}